\documentclass[runningheads]{waica}
\usepackage{amsmath,amsfonts}
\usepackage{algorithmic}
\usepackage{algorithm}
\usepackage{array}
\usepackage[caption=false,font=normalsize,labelfont=sf,textfont=sf]{subfig}
\usepackage{textcomp}
\usepackage{stfloats}
\usepackage{url}
\usepackage{verbatim}
\usepackage{graphicx}
\usepackage{cite}
\usepackage{enumitem}
\usepackage{xcolor}
\usepackage{multirow}
\usepackage{arydshln} 
\usepackage{adjustbox}
\usepackage{cleveref}
\usepackage{bbding}
\usepackage{booktabs}

\begin{document}

\title{Hierarchical Constrained Reinforcement Learning with Dynamic Boundary for Spatio-Temporal Vehicle-to-Grid Scheduling}

\titlerunning{HPC-RL}

\author{Haoyu Yan, Shutong Ding, Jiebao Zhang, Xi Yao, Yu Liu, Haoyu Wang, Chenchi Luo and Ye Shi
        <-this 
\thanks{(Corresponding authors: Ye Shi)}
\thanks{Haoyu Yan, Shutong Ding, Jiebao Zhang, Yu Liu, Haoyu Wang and Ye Shi are with the School of Information Science and Technology, ShanghaiTech University. (email:{\tt yanhy2023, dingsht, zhangjb2023, liuyu, wanghy, shiye@shanghaitech.edu.cn})}
\thanks{Xi Yao and Chenchi Luo are with China Mobile Shanghai ICT Co., Ltd. (email: {\tt yaoxi, luochenchi@cmsr.chinamobile.com})} 
}
\author{Haoyu Yan \inst{1} \and Shutong Ding \inst{1} \and Jiebao Zhang \inst{1} \and Xi Yao \inst{2} \and Yu Liu \inst{1} \and Haoyu Wang \inst{1} \and Chenchi Luo* \inst{2}\and Ye Shi*\inst{1}}
\footnotetext{*Chenchi Luo and Ye Shi are corresponding authors.}

\institute{The School of Information Science and Technology, ShanghaiTech University. \email{\tt yanhy2023, dingst,zhangjb2023, liuyu, wanghy, shiye@shanghaitech.edu.cn}\and
China Mobile Shanghai ICT Co., Ltd. \email{\tt yaoxi, luochenchi@cmsr.chinamobile.com} }
\authorrunning{Haoyu Yan Author et al.}
\maketitle


\begin{abstract}
The rapid proliferation of Electric Vehicles (EVs) introduces significant spatio-temporal uncertainties into power grids, while Vehicle-to-Grid (V2G) technology offers critical flexibility through bidirectional power flow. However, integrating large-scale EVs into the Optimal Power Flow framework presents substantial challenges due to computational bottlenecks arising from solver complexity and coupled spatio-temporal constraints. Existing Reinforcement Learning (RL) methods often struggle to balance strict constraint satisfaction with scalability in highly dynamic EV fleet environments. To address these challenges, this paper proposes a Hierarchical Policy for Constrained Reinforcement Learning (HPC-RL) framework for spatially and temporally coupled V2G scheduling. The framework adopts a two-layer architecture: the upper level utilizes a RL algorithm based on the Generalized Reduced Gradient method to strictly enforce spatial grid-level hard constraints; the lower level implements a novel dynamic boundary strategy to compute real-time feasible charging power bounds for individual EVs, thereby ensuring the satisfaction of temporal charging demands. This integrated design not only enables the simultaneous handling of spatially and temporally coupled constraints during the RL optimization process but also significantly enhances generalization capabilities for large-scale fleets through hierarchical decoupling. Extensive experiments on IEEE 14, 30, and modified 141-bus systems demonstrate that HPC-RL outperforms Model Predictive Control and state-of-the-art safe RL baselines across all metrics. The proposed method achieves near-optimal scheduling strategies and drastically reduces online computation time in large-scale scenarios from hours to minutes, while maintaining a near-zero constraint violation rate and nearly 100\% charging demand satisfaction.
\end{abstract}


\section{Introduction}
The accelerating global adoption of EVs introduces significant spatio-temporal uncertainties and coupled constraints into power grids. Nevertheless, EVs simultaneously offer unprecedented flexibility through V2G technology \cite{hartvigsson2022large, vagropoulos2016investigation}. Effectively leveraging large-scale EV fleets as distributed schedulable resources is critical for enhancing grid stability, mitigating renewable energy intermittency, and improving overall system efficiency. OPF serves as the fundamental tool for ensuring the secure, economical, and efficient operation of modern power systems \cite{lin2019non}. Consequently, seamlessly integrating V2G scheduling within the OPF framework is essential for maintaining grid stability.

However, optimizing V2G coordination within the OPF framework encounters high computational costs due to its Mixed-Integer Nonlinear Programming formulation \cite{tian2024two}. The complexity arises from discrete EV connection states, charging behaviors, and the nonlinearity of power flow equations. Traditional solvers, such as SCIP, scale poorly with the number of integer variables, rendering them impractical for the real-time coordination of large-scale fleets \cite{basu2022complexity, li2019route, spitzer2019optimized}. While MPC offers dynamic optimization under uncertainty, it suffers from significant computational overhead due to its iterative nature \cite{shi2018model}.

To mitigate these bottlenecks, research has shifted towards learning-based alternatives. Although Deep Neural Networks (DNNs) can approximate OPF solutions \cite{pan2021deepopffirst}, they often struggle with sequential decision-making. Deep Reinforcement Learning (DRL) is particularly well-suited for such time-evolving systems \cite{wu2024physics}. Despite the promise of DRL, two fundamental limitations persist in current methodologies. First, rigorous constraint enforcement remains an open challenge. Most existing Safe RL methods rely on soft penalties or post-hoc projections that do not guarantee strict adherence to physical power balance equality constraints. Second, scalability is often limited. Specifically, most frameworks assume a fixed number of EVs, failing to accommodate the dynamic fluctuations inherent in real-world V2G systems. 

To address these challenges, we propose the HPC-RL framework. RL serves as the computational backbone to enable real-time decision-making. Crucially, we embed the solution of power flow equations directly within the policy generation process, guaranteeing strict satisfaction of power balance equality constraints. Second, to satisfy time-coupled EV charging demands, we introduce a novel dynamic boundary strategy that calculates real-time feasible charging power ranges for individual EVs. Finally, we adopt a hierarchical architecture that decouples grid-level scheduling from individual EV dispatch, ensuring scalability.

The main contributions of this article are summarized as follows:
\begin{itemize}[leftmargin=4pt, rightmargin=4pt]
\item \textbf{Unified Enforcement of Spatio-Temporal Constraints in RL.} We integrate the GRG method into the RL framework, combined with a novel dynamic boundary strategy. Our approach rigorously enforces spatio-temporal constraints encompassing both grid-level physical safety and individual EV charging demands.
\item \textbf{A Hierarchical Framework for Scalable V2G Scheduling.} We propose a hierarchical architecture that decouples grid-level power dispatch from individual EV charging allocation. This design effectively overcomes the curse of dimensionality caused by large-scale EV fleets.
\item \textbf{Scalable Real-Time Scheduling with Rigorous Constraint Satisfaction.} Extensive experiments validate the framework's effectiveness. Compared with MPC, HPC-RL provides a much faster feasible scheduling alternative with 100\% demand satisfaction, while incurring a moderate objective-cost gap. Compared with baseline CRL algorithms, HPC-RL achieves stronger constraint satisfaction and demand fulfillment across all tested systems.
\end{itemize}

\section{Related Work} \label{sec:related_work}
Existing literature relevant to this study can be broadly categorized into algorithmic advancements in CRL and the application of learning-based methods in V2G scheduling.

\subsection{Constrained Reinforcement Learning}
In safety-critical power systems, ensuring strict adherence to operational constraints is paramount. Current CRL methodologies primarily adopt three strategies: 1) \textbf{Penalty Methods} \cite{cao2023physics, ye2022learning, sadeghianpourhamami2019definition}, which treat constraints as soft objectives and require meticulous reward shaping without guaranteeing strict safety; 2) \textbf{Optimization Embedding} \cite{wang2019safe, yan2022hybrid}, which dynamically balances objectives via Lagrangian multipliers but struggles with tuning dual variables for complex non-convex constraints; and 3) \textbf{Projection-based Strategies} \cite{kou2020safe, gao2022model, xia2022safe}, which correct infeasible actions post-hoc but often fail to strictly satisfy the non-linear power balance equality constraints inherent in grids.

\subsection{Learning-based Vehicle-to-Grid Scheduling}
Early V2G scheduling approaches employed DNNs to approximate OPF solutions \cite{pan2021deepopffirst, pan2020deepopfscdc, donti2020dc3}. However, these methods struggle with the sequential decision-making required for dynamic battery management \cite{li2021electric}. Consequently, RL has become the preferred choice for handling time-evolving V2G systems \cite{wu2024physics}. Despite this progress, most existing RL V2G frameworks are restricted to fixed-dimension state spaces, failing to accommodate the dynamic entry and exit of EVs \cite{chen2022deep, li2019constrained, wu2023network}. While some hierarchical designs attempt to manage variable EV numbers \cite{ting2024uncertainty, zhang2023ev}, they typically overlook the rigorous enforcement of hard physical constraints. 

As summarized in Table \ref{tab:comparison}, existing methodologies face a fundamental trade-off between real-time capability and the strict satisfaction of coupled spatio-temporal constraints. 
Specifically, standard CRL methods struggle to enforce the \textit{temporally coupled} charging demands of EVs as hard constraints, while projection-based strategies often violate critical grid equality constraints despite satisfying inequality limits. Furthermore, the inability to scale to variable EV fleet sizes limits the practical deployment of many RL solutions.

\begin{table*}[]
\centering
\footnotesize
\caption{Comparison among existing V2G schedule algorithms across various dimensions.}
\label{tab:comparison}
\scalebox{0.725}{
\begin{tabular}{c|ccccc}
\toprule
Method & \begin{tabular}[c]{@{}c@{}}Methodology \\ Category\end{tabular} & \begin{tabular}[c]{@{}c@{}}Real-time\\ requirement\end{tabular} & \begin{tabular}[c]{@{}c@{}}Network Security\\ (Spatial constraints)\end{tabular} & \begin{tabular}[c]{@{}c@{}}EV Task Guarantees\\ (Temporal constraints)\end{tabular} & Hard constraints \\ \hline
MPC \cite{li2022distributionally} & Optimization & \XSolidBrush & \Checkmark  & \Checkmark & \Checkmark \\
DC+trust region \cite{tian2024two} & Optimization & \XSolidBrush & \Checkmark  & \Checkmark & \Checkmark \\
DC3 \cite{donti2020dc3} & Learning & \Checkmark & \Checkmark & \Checkmark  & \Checkmark \\
GGAT+DAE \cite{cao2023physics} & Learning & \Checkmark & \Checkmark & \Checkmark  & \XSolidBrush \\
Penalty-based RL \cite{ye2022learning, sadeghianpourhamami2019definition} & Learning & \Checkmark & \XSolidBrush  & \Checkmark & \XSolidBrush \\
Opt-embedded RL \cite{wang2019safe, yan2022hybrid} & Learning & \Checkmark & \Checkmark  & \XSolidBrush & \Checkmark \\
Pro-based RL \cite{kou2020safe, gao2022model} & Learning & \Checkmark & \Checkmark  & \XSolidBrush & \Checkmark \\
RPO \cite{ding2023reduced} & Learning & \Checkmark & \Checkmark & \XSolidBrush & \Checkmark \\
RDDPG with DEB \cite{li2023constrained} & Learning & \Checkmark & \XSolidBrush  & \Checkmark & \XSolidBrush \\
SAC+NMT \cite{li2023constrained} & Learning & \Checkmark & \Checkmark  & \Checkmark & \XSolidBrush \\
CRL \cite{wu2023network} & Learning & \Checkmark & \Checkmark  & \Checkmark & \XSolidBrush \\
\hline
HPC-RL (Ours) & Learning & \Checkmark & \Checkmark  & \Checkmark & \Checkmark \\
\bottomrule
\end{tabular}
}
\end{table*}

\begin{table*}[]
\centering
\footnotesize
\caption{Comparison among existing V2G schedule algorithms across various dimensions.}
\label{tab:comparison}
\scalebox{0.725}{
\begin{tabular}{c|ccccc}
\toprule
Method & \begin{tabular}[c]{@{}c@{}}Methodology \\ Category\end{tabular} & \begin{tabular}[c]{@{}c@{}}Real-time\\ requirement\end{tabular} & \begin{tabular}[c]{@{}c@{}}Network Security\\ (Spatial constraints)\end{tabular} & \begin{tabular}[c]{@{}c@{}}EV Task Guarantees\\ (Temporal constraints)\end{tabular} & Hard constraints \\ \hline
MPC \cite{li2022distributionally} & Optimization & \XSolidBrush & \Checkmark  & \Checkmark & \Checkmark \\
Penalty-based RL \cite{ye2022learning, sadeghianpourhamami2019definition} & Learning & \Checkmark & \XSolidBrush  & \Checkmark & \XSolidBrush \\
Opt-embedded RL \cite{wang2019safe, yan2022hybrid} & Learning & \Checkmark & \Checkmark  & \XSolidBrush & \Checkmark \\
Pro-based RL \cite{kou2020safe, gao2022model} & Learning & \Checkmark & \Checkmark  & \XSolidBrush & \Checkmark \\
\hline
HPC-RL (Ours) & Learning & \Checkmark & \Checkmark  & \Checkmark & \Checkmark \\
\bottomrule
\end{tabular}
}
\end{table*}

These deficiencies highlight a critical gap: existing methodologies face a fundamental trade-off between real-time scalability for dynamic EV fleets and the strict satisfaction of coupled spatio-temporal constraints. The proposed HPC-RL framework addresses this by integrating a full AC-OPF formulation with a hierarchical architecture.

\section{Problem Formulation}
\label{sec:problem_formulation}

\subsection{OPF-EV Formulation}
V2G scheduling must satisfy AC power-flow feasibility at every dispatch step while ensuring that each EV receives sufficient energy before departure. This motivates a constrained sequential OPF-EV formulation that minimizes the combined grid-side generation cost and EV charging cost under both network security constraints and EV service guarantees:
\begin{align}
    \underset{p_g,q_g,v,p_{ch}}{\text{min }}  &\sum_{t=0}^T p_g^\top(t) A p_g(t) + b^\top p_g(t) + c(t) P_{ch}(t) ,\label{eq:obj} \\
    \text{subject to } &(p_g(t)-p_d(t)-P_{ch}(t)) + \mathrm{i}(q_g(t)-q_d(t)) \notag \\
    & = \text{diag}(v(t)) \, Y^* \, v(t)^*, \label{eq:power_balance} \\
    & {\underline{p}}_g \le p_g \le {\overline{p}}_g ,\label{eq:pg_bound} \\
    & {\underline{q}}_g \le q_g \le {\overline{q}}_g ,\label{eq:qg_bound} \\
    & {\underline{v}} \le |v(t)| \le {\overline{v}} ,\label{eq:v_bound} \\
    & {\underline{P}}_{ch} \le P_{ch} \le {\overline{P}}_{ch} ,\label{eq:pb_bound} \\
    & P_{ch}(t) = \sum_i P_{i, ch}(t) ,\\
    & SOC_i(t) = SOC_i(t - 1) + \eta P_{i, ch}(t) \Delta t ,\label{evstart}\\
    & SOC_i^{min} \le SOC_i(t) \le SOC_i^{max}, \\
    & P_i^{min} \le P_{i, ch}(t) \le P_i^{max} ,\\
    & SOC_i^l \ge SOC_i^d .\label{evend}
\end{align}
Here, $p_g, q_g \in \mathbb{R}^n$ represent the active and reactive power generation at each bus, and $v \in \mathbb{C}^n$ represents the bus voltages in the grid. $Y \in \mathbb{C}^{n \times n}$ is the admittance matrix, while $p_d, q_d \in \mathbb{R}^n$ denote the active and reactive power demands. Not all buses are generator buses; at non-generator buses, $p_g$ and $q_g$ are zero. One generator bus serves as the reference (slack) bus, maintaining a fixed voltage angle $\angle v$.
$P_{ch} \in \mathbb{R}^n$ represents the charging power of EVs, and $c(t)$ is the cost or revenue at time step $t$. Equation (\ref{eq:power_balance}) enforces power balance. Equations (\ref{eq:pg_bound}) and (\ref{eq:qg_bound}) are generator capacity constraints. Equation (\ref{eq:v_bound}) ensures voltage security. Equation (\ref{eq:pb_bound}) restricts EV charging rates. Equations (\ref{evstart})-(\ref{evend}) define the SOC dynamics and ensure EV energy demand is met.

\subsection{Markov Decision Process}
We formulate the OPF-EV scheduling as a finite-horizon Markov Decision Process (MDP) defined by the tuple $\{ \mathcal{S, A, P, R, \gamma} \}$:

\textbf{State ($\mathcal{S}$) \& Action ($\mathcal{A}$):} $\mathcal{S} = \{ p_d, q_d, SOC, r_t, c(t) \}$ encodes the current grid demands and active EV information, including battery states, remaining connection times ($r_t$), target charging requirements, charging-rate limits, and real-time electricity prices. $\mathcal{A} = \{ p_g, q_g, v, \theta, P_{i,ch} \}$ encompasses both grid variables (to satisfy power flow equations) and non-negative EV charging rates.

\textbf{Transition ($\mathcal{P}$) \& Reward ($\mathcal{R}$):} Since grid demands are exogenous, state transitions primarily follow deterministic battery dynamics: $SOC_i^{t+1} = SOC_i^t + \eta P_{i,ch}^t \Delta t$. The reward $\mathcal{R}$ captures the dual economic objectives of minimizing the grid's quadratic generation cost and the EV charging station's electricity procurement cost:
\begin{equation}
    \mathcal{R} = -w_{\text{grid}} \sum_{t=0}^{T} ( p_g^\top A p_g + b^\top p_g ) - w_{\text{ev}} \sum_{t=0}^{T} c(t) P_{\text{ch}}(t). \label{eq:total_reward}
\end{equation}

Solving this MDP presents three fundamental hurdles: 1) \textit{Action Feasibility:} Conventional soft penalties fail to strictly enforce the complex non-linear power balance equality constraints; 2) \textit{Dynamic Dimensionality:} The continuous arrival and departure of EVs create a varying action space size, degrading standard RL network stability \cite{klissarov2025discovering}; 3) \textit{Temporally Coupled Constraints:} EV charging demands are only verifiable at departure, preventing standard step-wise RL from enforcing them during ongoing scheduling.

\section{The Proposed Model and Solution Approaches}

\begin{figure*}[h]
    \centering
    \includegraphics[width=0.8\linewidth]{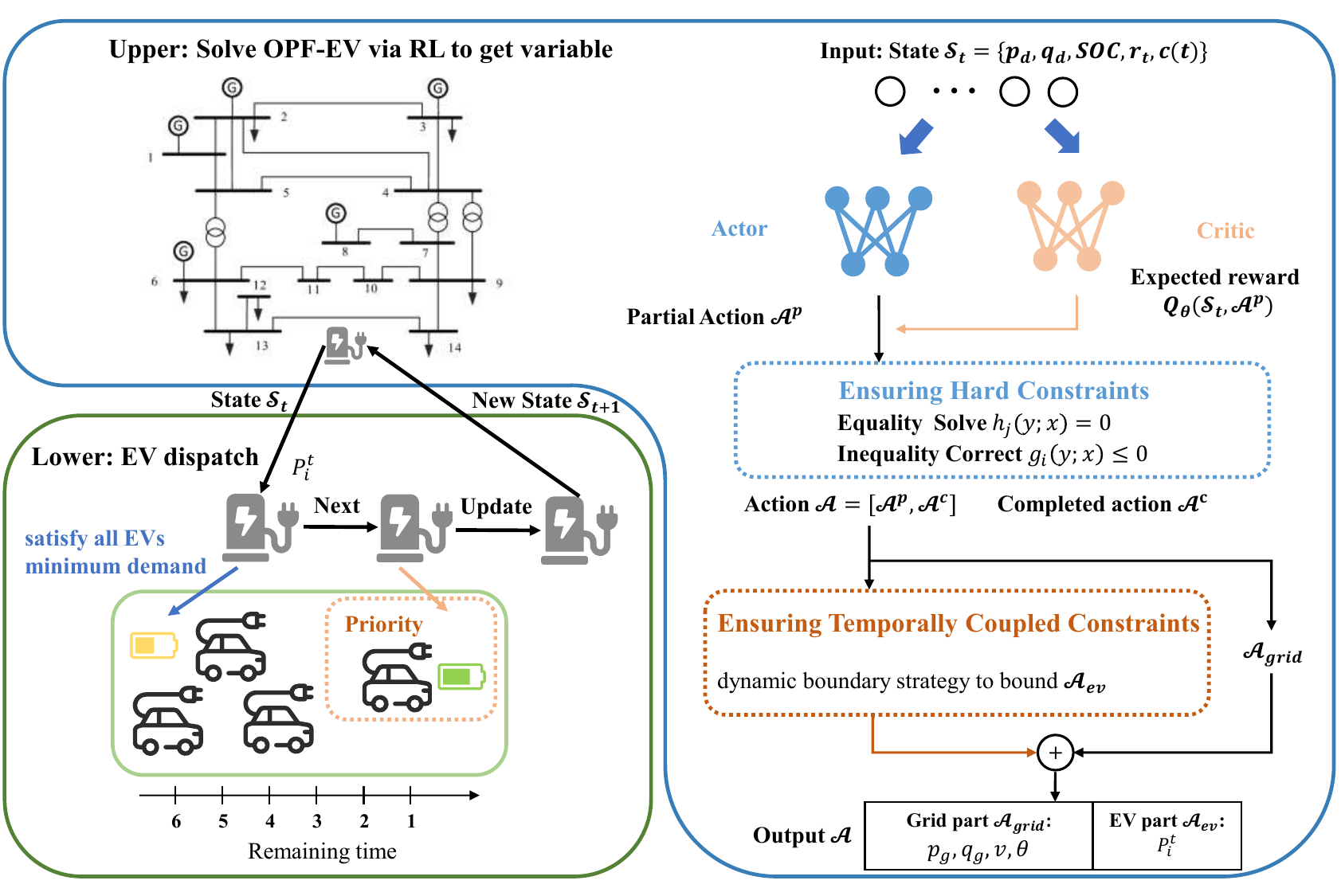}
    \caption{The HPC-RL pipeline consists of an upper layer solving the OPF-EV using RL with GRG constraints, and a lower layer managing EV charging allocations via dynamic boundaries and priority.}
    \label{fig:pipeline}
\end{figure*}

To address the OPF-EV problem, we propose a hierarchical policy framework utilizing a boundary strategy to decouple scheduling tasks across time and space which shown in Fig.\ref{fig:pipeline}. The upper level employs a reinforcement learning network, integrated with equation-solving and inequality-correcting layers, to optimize grid parameters while strictly satisfying spatial constraints. Meanwhile, the lower level allocates power to guarantee minimum EV charging requirements and prioritize remaining capacity, ensuring all temporally coupled constraints are met.

\subsection{Hierarchical Decomposition via GRG-SAC}
The hierarchical framework is structured according to power grid topology, aggregating EVs arriving at the same distribution node into electric vehicle charging stations (EVCSs). This configuration mirrors real-world operations where distribution nodes service clusters of EVCSs. The original optimization problem is decoupled into two layers: the upper level coordinates grid-level power generation variables ($\mathbf{p_g, q_g, v, \theta}$) with the aggregate charging demands of the EVCSs ($\mathbf{P_{ch}}$), while the lower level executes power allocation per EV within each station. Through this hierarchical decoupling based on physical topology, we simultaneously address the curse of dimensionality while preserving the physical fidelity of the grid model.
Numerical solvers are computationally prohibitive for real-time control, whereas standard RL struggles to strictly enforce power flow equations. We adopt a physics-integrated GRG-SAC framework combining GRG with SAC, which analytically embeds power flow equality constraints into the action generation process. Rather than treating the optimization as a generic constrained problem, we analytically embed the power flow equations directly into the action generation process. We decompose the decision variables into two specific sets: \textbf{basic actions} ($a_B$) and \textbf{non-basic variables} ($a_N$). The RL policy $\mu_\theta(s)$ is restricted to outputting only the independent basic actions $a_B$, specifically the active power generation $\mathbf{p_g}$ at generator buses and voltage magnitudes $\mathbf{|v|}$ at non-load buses. Subsequently, the non-basic variables $a_N$, comprising voltage magnitudes at load buses and phase angles $\theta$, are analytically derived by solving the power balance equality constraints $h(a_B, a_N) = 0$ via Newton's method. This equation-solving procedure effectively establishes an implicit mapping $a_N = \phi(a_B)$, where the dependent variables are deterministically constrained by the physics of the grid.

To guarantee the validity of this decomposition during training, we utilize the Implicit Function Theorem to propagate gradients through this implicit mapping. The gradient flow from the non-basic variables back to the policy output is derived as:
\begin{equation}
    \frac{\partial a_N}{\partial a_B} = - \left( \frac{\partial h}{\partial a_N} \right)^{-1} \frac{\partial h}{\partial a_B} = -(J_{N})^{-1} J_{B},
\end{equation}
where $J_N$ and $J_B$ are partitions of the power flow Jacobian matrix. To handle inequality constraints $g(a) \leq 0$, we employ a projection mechanism that maps the action generated by the policy back to the feasible set whenever an inequality violation occurs. The physical feasibility of this correction is rigorously guaranteed by the following proposition:

\begin{proposition}[Feasibility Preservation in Tangent Space]
The implicit adjustment $\Delta a_N$ confines the total update vector $\Delta a = [\Delta a_B^T, \Delta a_N^T]^T$ to the tangent space of the equality constraints, strictly preserving the power balance manifold.
\end{proposition}
\begin{proof}
Let $h(a)=0$ denote the power-flow equality constraints and $a=[a_B^T,a_N^T]^T$. Denote $J_B=\partial h/\partial a_B$ and $J_N=\partial h/\partial a_N$. If $J_N$ is nonsingular, the GRG adjustment induced by a perturbation $\Delta a_B$ is
\begin{equation}
    \Delta a_N = -J_N^{-1}J_B\Delta a_B .
\end{equation}
Therefore,
\begin{equation}
    \nabla h(a)\Delta a = J_B\Delta a_B + J_N\Delta a_N = 0,
\end{equation}
which shows that the update direction lies in the tangent space of the equality-feasible manifold. By the implicit function theorem, a local mapping $a_N=\phi(a_B)$ exists such that $h(a_B,\phi(a_B))=0$. In implementation, Newton iterations are used to recover nonlinear equality feasibility up to numerical tolerance.
\end{proof}
This establishes first-order tangent-space feasibility for the equality constraints. The subsequent Newton correction recovers nonlinear AC power-flow feasibility up to the solver tolerance, ensuring that the executed action remains close to the equality-feasible manifold.

The GRG-SAC hybrid struggles with time-coupled EV charging requirements due to discrete state transitions and multi-step temporal dependencies that contradict GRG's single-step optimization. Consequently, the core challenge lies in converting long-term charging obligations into real-time, actionable boundaries to guide hierarchical scheduling. To bridge this gap, the following subsection introduces a novel boundary strategy.

\subsection{Dynamic Boundary Strategy for Temporally Decoupling EV Constraints}
In this paper, we introduce a novel dynamic boundary strategy that translates long-term charging obligations into real-time operational boundaries. This approach bridges the temporal decoupling between instantaneous grid optimization and multi-step EV charging objectives while accommodating discrete state transitions induced by EV arrival and departure events.

The core innovation resides in dynamically computing feasible power boundaries that guarantee each EV $i$ satisfies its demand. For individual EV $i$ at time $t$, we derive these boundaries through temporal-aware projections:

\begin{equation}
    \label{eq:boundary}
    \scalebox{0.95}{$\displaystyle
    \begin{split}
        &P_{i,ch}^{\max}(t) = \min \{P_{i,ch}^{\text{phys}}, SOC^{\max} - SOC_i(t) \}, \\
        &P_{i,ch}^{\min}(t) = \max \{ 0, (SOC^d_i - SOC_i(t)) - d_{\text{re}}(t) \cdot P_{i,ch}^{\text{phys}} \}, 
    \end{split}
    $}
\end{equation}
where $d_{\text{re}}(t) = t_{\text{dep}} - t - 1$ represents the remaining decision intervals excluding the current time step. 

The formulation in \eqref{eq:boundary} comprehensively characterizes the dynamic determination of power constraint boundaries. The upper bound $P_{i,ch}^{\max}$ enforces dual limitations from maximum charging rate $P_{i,ch}^{\text{phys}}$ and residual energy capacity $(SOC^{\max} - SOC_i(t))$, ensuring adherence to both hardware limitations and battery operational boundaries. The lower bound $P_{i,ch}^{\min}$ selects the maximum between zero (minimum charging rate) and the demand-driven power requirement, which guarantees charging completion within the remaining time window.

The dynamic boundary strategy decouples long-term obligations from real-time optimization, mathematically guaranteeing that all EVs meet their charging requirements while maintaining operational feasibility. However, relying solely on upper-level optimization introduces operational flaws, as it cannot guarantee individual minimum SOC limits or prioritize urgent charging needs. To address these limitations, a lower-level allocation mechanism is necessary to manage individual EVs. This lower-level allocation first ensures that all EVs receive their minimum required power, enforced through boundary constraints. Finally, the remaining station capacity is dynamically distributed to EVs with the highest charging urgency, determined by their SOC deficits and imminent departure times.

For individual EV energy management, we compute the required charge quantities to determine aggregate EVCS states. The RL module outputs normalized charging solutions constrained within feasible operational spaces using the hyperbolic tangent function. Furthermore, anticipating temporal variations in charging demand, we design a specialized demand embedding layer that encodes future requirements into the state space as follows:
\begin{equation}
    s_k^{\text{demand}} = \sum_{\tau=k}^{T} d_\tau \quad \forall k \in \{1,2,\dots,T\},
\end{equation}
where $\mathbf{d} = [d_1, d_2, \dots, d_T]$ is a $T$-period demand vector. For instance, a 6-hour demand profile $[0, 0.2, 0.2, 0.2, 0, 0.2]$ generates the encoded state $[0.8, 0.8, 0.6, 0.4, \allowbreak 0.4, 0.2]$. The advantages of this representation will be demonstrated in the subsequent experimental section.

\subsection{HPC-RL Framework Summary}
In summary, to holistically address the aforementioned challenges, we propose the HPC-RL framework. This framework is architected around two core design principles. A hierarchical decomposition that decouples grid-level optimization from EV-level allocation to tackle scalability and dynamic dimensionality and a hybrid constraint handling strategy to guarantee strict feasibility. The upper-level agent embeds GRG's constraint satisfaction within SAC's exploratory policy optimization, ensuring all grid physical constraints while generating actions. Concurrently, our dynamic boundary strategy transforms long-term, temporally coupled EV charging obligations into instantaneous feasible charging limits. This design fundamentally resolves the challenges by guaranteeing physically feasible actions and converting discontinuous EV connection events into continuously differentiable boundary conditions for stable learning.

\section{Experiments}
\subsection{Experiments Settings}
\subsubsection{Data set and parameters}
We conducted experiments on the Case 14 system, Case 30 system and a modified Case 141 system (which selected 52 nodes as new power generation nodes). For the Case 14 system, three EVCSs are deployed at buses 2, 6, and 8; for the Case 30 system, three EVCSs are deployed at buses 1, 13, and 27; and for the modified Case 141 system, a single EVCS is implemented at the newly added generation bus 52, with each EVCS maintaining a constant connection capacity of 10 EVs per discrete time step $T$. 

The scheduling horizon is $T=24$ hours, and the time granularity is $\Delta t=1$ hour. For each EVCS, a vehicle will arrive in each time step $T$, and the parking time of each vehicle is fixed to 8 hours. For each EV, the charging efficiency is set to $\eta=0.98$, the maximum charging rate is $P_i^{max}=0.2$, the arrival SOC is $SOC_i=0.2$, and the target SOC is $SOC_i^d=0.8$. PyTorch is used as the primary simulation environment and MATLAB's CVX solver is used for numerical baselines. The fixed EV profiles are used only as controlled evaluation trajectories to ensure fair and reproducible comparison across methods. HPC-RL itself operates online and does not require future EV arrivals, dwell times, or demands; at each step, it only uses the current grid state and the active EV information.

\subsubsection{Benchmarks}
We compare HPC-RL against traditional numerical optimization using MPC \cite{li2022distributionally}, as well as several CRL baselines: Constrained Policy Optimization (CPO) \cite{achiam2017constrained}, Conservative Update Policy (CUP) \cite{yang2022cup}, DDPGLA \cite{silver2014deterministic}, and SACLA \cite{haarnoja2018soft}.

\subsubsection{Metrics}
\begin{itemize}[leftmargin=4pt, rightmargin=4pt]
    \item Objective function: Equations (\ref{eq:obj}) represents the cost of power generation for the grid and the cost of charging EVs.
    \item Constraint Violation: Quantifies the maximum deviation from both equality and inequality constraints. Higher values indicate greater infeasibility.
    \item Demand Satisfaction Rate: Evaluates the percentage of EVs that successfully achieve their demand before departure from the EVCS.
\end{itemize}

\subsection{Results}

\begin{figure}[h]
    \centering
    \includegraphics[width=1.0\linewidth]{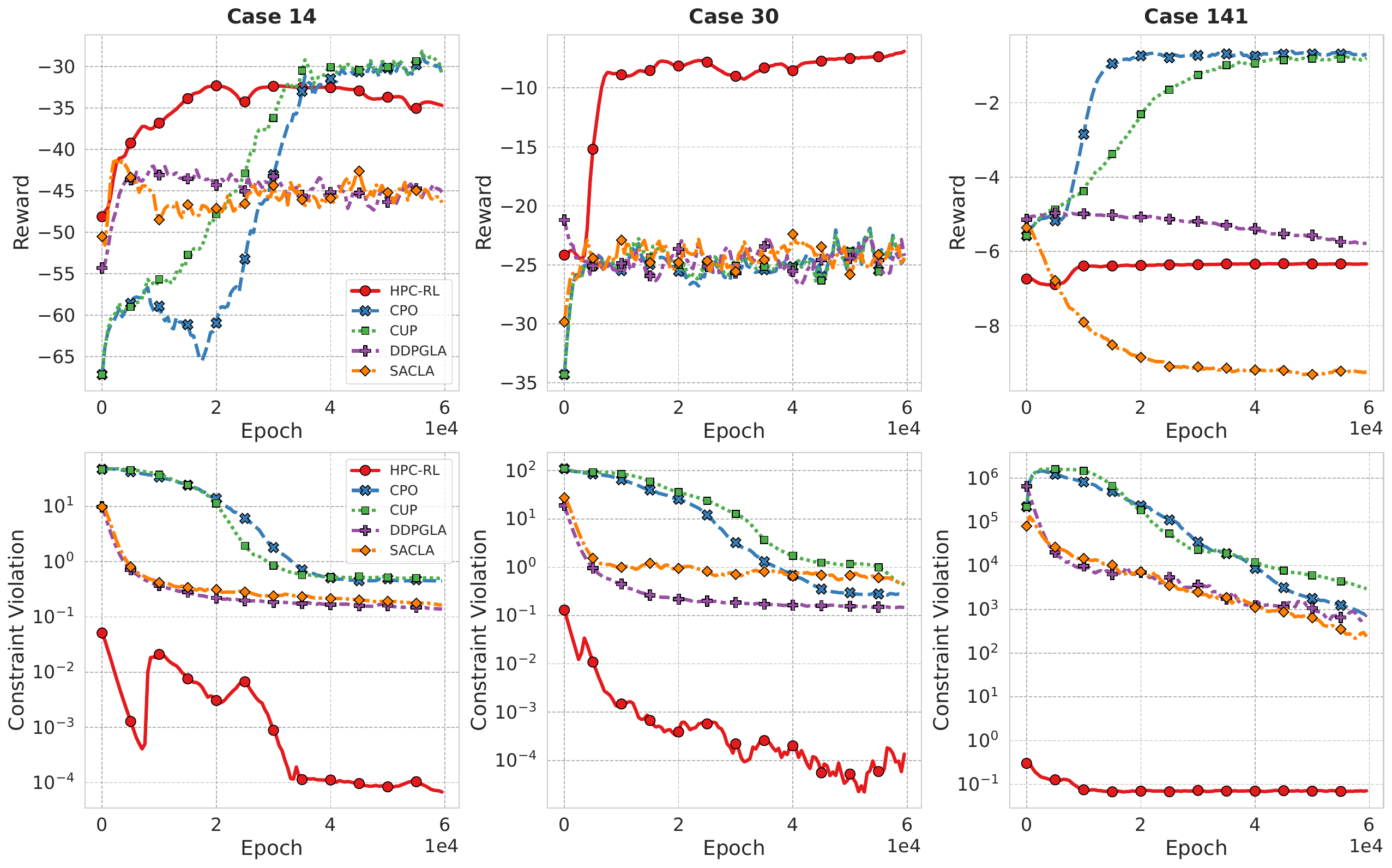}
    \caption{This figure shows the learning curves of five different safe reinforcement learning methods (HPC-RL, CPO, CUP, DDPGLA, and SACLA) in three different scenarios (Case 14, Case 30, and Modify Case 141). The x-axis represents the number of training epochs. The y-axis represents the episodic reward (first row), and the maximum instantaneous constraint violation (second row).} 
    \label{fig:case}
\end{figure}

\textbf{Comparative analysis with Safe-RL algorithms.} Figure \ref{fig:case} illustrates the learning curves across the three scenarios. Across all cases, baseline algorithms exhibit critical trade-offs: CPO and CUP often achieve high episodic rewards but suffer from severe constraint violations and poor demand satisfaction, rendering them physically unfeasible for grid operations. Conversely, DDPGLA and SACLA maintain better safety metrics but fail to achieve optimal reward convergence. Our proposed HPC-RL consistently resolves this trade-off among RL baselines by achieving stable rewards, near-zero constraint violations, and 100\% demand satisfaction across all tested network scales.

\textbf{Computational Efficiency and Optimality Preservation.} As demonstrated in Table~\ref{tab1}, MPC achieves the lowest objective values among feasible methods, while HPC-RL provides a much faster feasible alternative with objective gaps of about 10.2\%, 13.1\%, and 7.2\% for Cases 14, 30, and 141, respectively, and reduces runtime by factors of $52\times$, $67\times$, and $319\times$. For large-scale networks where the computational burden of MPC limits online scheduling, HPC-RL maintains low online inference time after training, supporting its deployment as a real-time feasible scheduling method.

\begin{table*}[h]
\centering
\caption{Average evaluation performance across three benchmarks comparing numerical solvers with Ours and safe RL baselines.}
\scalebox{0.8}{
\begin{tabular}{l|c|lccccc}
\hline
Case & Category   & Method  & Obj.    & Runtime (s)  & Ineq (viol) & Eq (viol) & Demand Satisfaction (\%)\\ 
\hline
\multirow{6}{*}{14} 
& \multirow{1}{*}{Traditional}  
    & MPC     & 34.2143 & 308.550  & 0.00       & 0.0000     &100.0   \\
  \cline{2-8} 
& \multirow{5}{*}{\centering RL-based}  
    & CPO     & 20.8081 & 0.005  & 0.00       & 0.7710   &68.6 \\
&   & CUP     & 23.7842 & 0.005  & 0.00       & 0.7545   &66.2  \\ 
&   & DDPGLA  & 48.0469 & 0.007  & 0.00       & 0.2322   &72.3 \\ 
&   & SACLA   & 44.7153 & 0.015  & 0.00       & 0.2321   &73.9  \\ 
&   & Ours    & 37.6995 & 1.016  & 0.00       & 0.0020   &100.0   \\ 
\hline 
\multirow{6}{*}{30} 
& \multirow{1}{*}{Traditional}  
    & MPC     & 7.1054  & 473.470  & 0.00       & 0.0000     &100.0   \\
  \cline{2-8} 
& \multirow{5}{*}{\centering RL-based}  
    & CPO     & 25.1416 & 0.005  & 0.00       & 0.3873   &78.2  \\
&   & CUP     & 25.1419 & 0.005  & 0.00       & 0.3873   &77.7  \\ 
&   & DDPGLA  & 23.3126 & 0.005  & 0.00       & 0.1681   &71.9  \\ 
&   & SACLA   & 22.7059 & 0.016  & 0.00       & 0.7478   &70.8  \\ 
&   & Ours    & 8.0428  & 0.872  & 0.00       & 0.0002   &100.0   \\ 
\hline
\multirow{6}{*}{141}
& \multirow{1}{*}{Traditional}  
    & MPC     & 5.1221  & 5172.700  & 0.00       & 0.0000     &100.0   \\
 \cline{2-8} 
& \multirow{5}{*}{\centering RL-based}  
    & CPO     & 0.7177  & 0.005  & 0.00       & 164.25   &75.2  \\
&   & CUP     & 0.8858  & 0.005  & 0.00       & 1706.1   &78.9  \\ 
&   & DDPGLA  & 5.7690  & 0.005  & 0.00       & 661.73   &70.2  \\ 
&   & SACLA   & 9.0711  & 0.014  & 0.00       & 628.50   &92.2  \\ 
&   & Ours    & 5.4945  & 4.967  & 0.00       & 0.4138   &100.0   \\ 
\hline
\end{tabular}
}
\label{tab1}
\end{table*}

\textbf{Learned Price-Responsive Behavior.} Figure \ref{fig:price} demonstrates that HPC-RL effectively learns to arbitrage dynamic electricity prices. The system aggressively charges EVs during low-price periods and minimizes charging during price spikes. Crucially, the EVCS still strictly enforces minimum charging boundaries during high-price periods to ensure all vehicle energy demands are fulfilled prior to departure.

\begin{figure}[h]
    \centering
    \includegraphics[width=0.7\linewidth]{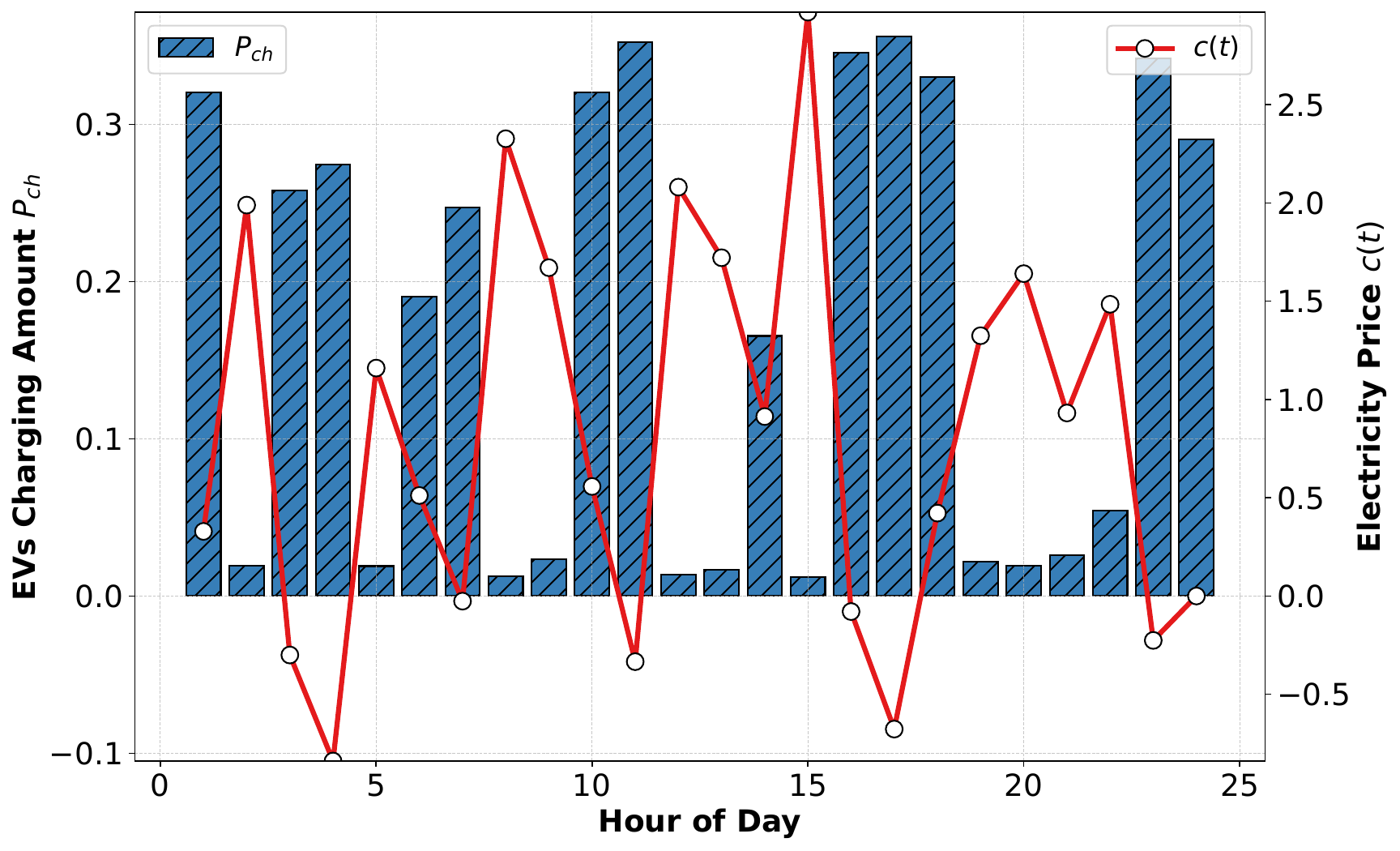}
    \caption{Relationship between electricity price and charging amount.}
    \label{fig:price}
\end{figure}

\subsection{Scalability}
In this section, we present the experimental results of an experiment designed to evaluate the scalability of our architecture, focusing on Case 14. We tested the system's performance with varying numbers of electric vehicles connected to a single charging station, considering scenarios with 1, 10, and 50 electric vehicles, to assess how the system responds to different levels of demand. The recorded execution times for each scenario are shown in Table \ref{Sca}. These results demonstrate that although the execution time increases with the number of electric vehicles, our architecture is still able to effectively manage charging demands, while the generated solutions are very close to those obtained by the numerical solver. This highlights its strong scalability and efficiency in adapting to different numbers of connected electric vehicles, thus confirming its suitability for practical applications in smart grid environments. The EV numbers in Table~\ref{Sca} represent instantaneous active-set sizes used to evaluate per-step inference cost. In online scheduling, EVs may enter and leave over time, and HPC-RL constructs the lower-level allocation according to the current active EV set at each decision step.

\begin{table}[h]
\centering
\caption{Comparison of HPC-RL and MPC with varying numbers of EVs.}
\scalebox{0.9}{
\begin{tabular}{ll|ll}
\hline
\multicolumn{2}{l|}{IEEE Case 14}       & MPC     & Ours    \\ \hline
\multicolumn{1}{l|}{}      & Obj.       & 19.9023 & 20.3576 \\
\multicolumn{1}{l|}{1 EV}  & Runtime(s) & 100.79  & 0.9991  \\
\multicolumn{1}{l|}{}      & viol       & 0.00    & 0.00    \\ \hline
\multicolumn{1}{l|}{}      & Obj.       & 23.1174 & 23.8678 \\
\multicolumn{1}{l|}{10 EVs}& Runtime(s) & 224.52  & 1.0033  \\
\multicolumn{1}{l|}{}      & viol       & 0.00    & 0.00     \\ \hline
\multicolumn{1}{l|}{}      & Obj.       & 39.2993 & 39.8084 \\
\multicolumn{1}{l|}{50 EVs}& Runtime(s) & 1020.28 & 1.9313  \\
\multicolumn{1}{l|}{}      & viol       & 0.00    & 0.00    \\ \hline
\end{tabular}
}
\label{Sca}
\end{table}

\subsection{Ablation Study}
\textbf{Impact of Demand Information Embedding.} Figure \ref{fig:with demand} highlights that incorporating future EV charging demand information into the state representation significantly improves learning efficiency. The demand-aware model not only converges to higher rewards but also exhibits reduced volatility (narrower standard deviation bands) and faster suppression of constraint violations compared to the baseline.

\begin{figure}[h]
    \centering
    \includegraphics[width=0.8\linewidth]{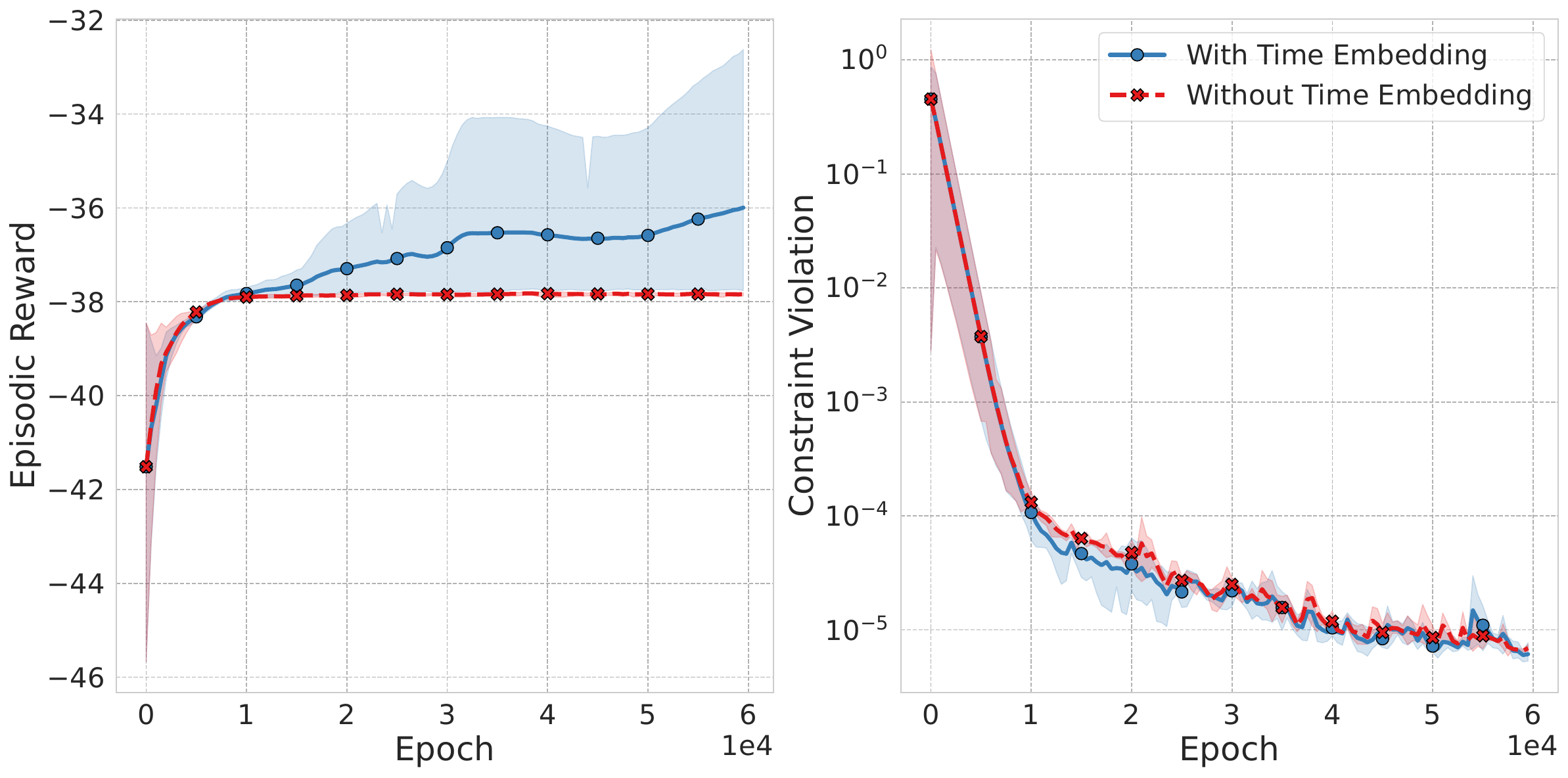}
    \caption{The influence of demand information embedding.}
    \label{fig:with demand}
\end{figure}

\textbf{Evaluation of Lower-Level Resource Allocation.} Figure \ref{fig:lower-level} compares our dynamic priority-based allocation against a baseline average distribution method. While both maintain constraint violations within a negligible $10^{-4}$ margin, our approach yields vastly superior reward accumulation. By dynamically adjusting resources based on instantaneous EV charging urgency rather than static uniform distribution, the proposed strategy ensures optimal EVCS capacity utilization during demand fluctuations.

\begin{figure}[h]
    \centering
    \includegraphics[width=0.8\linewidth]{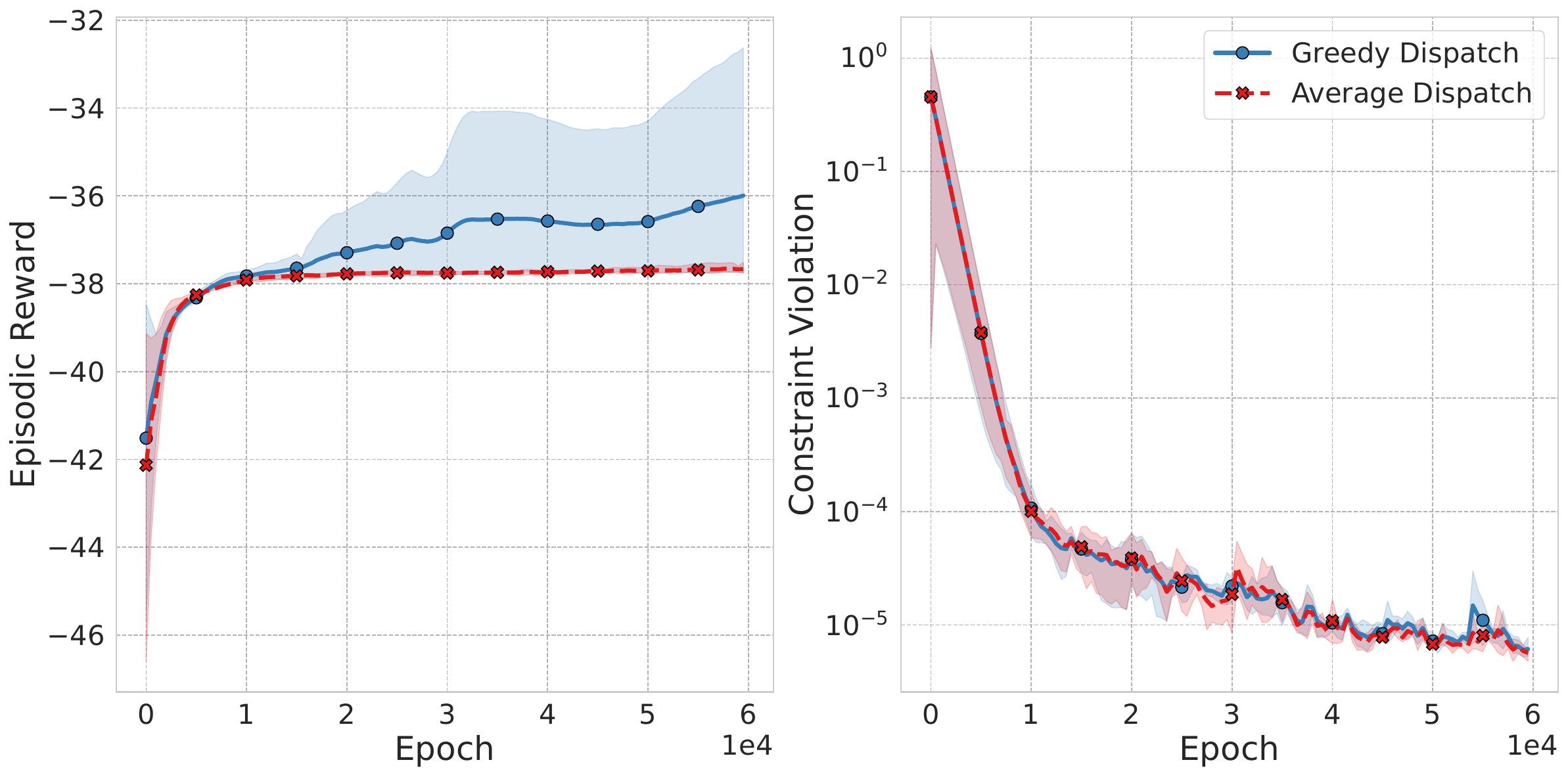}
    \caption{Comparison between Priority-Based Allocation and Average Allocation.}
    \label{fig:lower-level}
\end{figure}

\section{CONCLUSION}
In this paper, we propose the HPC-RL framework, a physics-integrated hierarchical solution designed to address the complex coupling of spatial grid constraints and temporal EV charging demands in V2G scheduling. By analytically embedding the power flow equations into policy optimization via the GRG method, our approach keeps policy updates within the tangent space of the power-flow manifold and recovers nonlinear equality feasibility up to numerical tolerance through Newton correction. This mechanism reduces reliance on soft penalty terms and improves the enforcement of hard physical constraints.

Furthermore, the proposed Dynamic Boundary Strategy successfully transforms long-term, temporally coupled charging obligations into instantaneous feasible operational regions, effectively bridging the gap between real-time grid dispatch and sequential battery management. Experimental evaluations on IEEE 14, 30, and modified 141-bus systems demonstrate that HPC-RL achieves a practical trade-off among feasibility, demand satisfaction, and online efficiency. It substantially outperforms state-of-the-art RL baselines in constraint satisfaction and demand fulfillment, while serving as a much faster feasible alternative to MPC with a moderate objective-cost gap. These results validate HPC-RL as a computationally tractable and physically reliable paradigm for large-scale V2G coordination. Future work will extend this framework to stochastic renewable generation, where renewable outputs can be incorporated as time-varying states or forecast inputs, and to more complex distribution network topologies.

\section*{Acknowledgment}
This work was supported by the National Natural Science Foundation of China (62303319, 62406195), HPC Platform of ShanghaiTech University, and Key Laboratory of Intelligent Perception and Human-Machine Collaboration (Shanghaitech University), Ministry of Education.

\bibliographystyle{IEEEtran}
\bibliography{reference}

\vfill
\end{document}